\documentclass[12pt]{article}

\usepackage{amsmath, amsthm, amssymb, stmaryrd, enumitem}

\newtheorem{theorem}{Theorem}[section]

\newtheorem{lemma}[theorem]{Lemma}

\theoremstyle{definition}

\begin{document}

\title{Complexity of the CNF-satisfiability problem}

\author{Grigoriy V. Bokov\\
        \small Lomonosov Moscow State University\\
        \small E-mail: bokov@intsys.msu.ru}

\maketitle

\begin{abstract}
This paper is devoted to the complexity of the Boolean satisfiability problem. We consider a version of this problem, where the Boolean formula is specified in the conjunctive normal form. We prove an unexpected result that the CNF-satisfiability problem can be solved by a deterministic Turing machine in polynomial time.
\end{abstract}



\section{Introduction} \label{S:1}

The Boolean satisfiability problem is a classic combinatorial problem. Given a Boolean formula, this problem is to check whether it is satisfiable, i.e., to find such values of the variables, on which the formula takes on the true value. There are several special cases of the Boolean satisfiability problem in which the formulas are restricted to a particular structure. The CNF-satisfiability problem is a version of the Boolean satisfiability problem, where the Boolean formula is specified in the conjunctive normal form (CNF), i.e., a conjunction of clauses, where each clause is a disjunction of literals, and a literal is a variable or its negation.

Cook~\cite{Coo71} and independently Levin~\cite{Lev73} proved that the CNF-satisfiability problem is $\mathbf{NP}$-complete. This proof shows how every decision problem in the complexity class $\mathbf{NP}$ can be reduced to the Boolean satisfiability problem for formulas in the conjunctive normal form. At the same time, for some restrictions such that 2-CNF-satisfiability or Horn-satisfiability, this problem can be solved in polynomial time or even in linear time~\cite{Kro67,DG84}. This also holds true for the evaluation problem of quantified Boolean formulas~\cite{APT79,BKF95}. All restrictions are generalized by Schaefer's dichotomy theorem that states necessary and sufficient conditions under which the Boolean satisfiability problem for a given restriction to Boolean functions is in $\mathbf{P}$ or $\mathbf{NP}$-complete~\cite{Sch78}.

The aim of this paper is to prove that the CNF-satisfiability problem without restrictions to Boolean formulas is polynomial-time decidable also. We describe a simple polynomial-time algorithm that, for a given Boolean formula in CNF, stops and answers yes if the formula is satisfiable and no otherwise. Besides, if the formula is satisfiable, the algorithm returns its true assignment.



\section{Preliminaries and results} \label{S:2}

\emph{Boolean formulas} are built up from the signature $\{\neg, \wedge, \vee\}$, variables and parentheses in the usual way. We use letters $x, y, z$, etc., to denote variables and literals, and capital letters $A$, $B$, $F$, $G$, etc., to denote propositional formulas.

A \emph{literal} is a variable $x$ or its negation $\overline{x}$ and denoted by $x^a$, where $x^1 = x$ and $x^0 = \overline{x}$. We assume that $\overline{\overline{x}} = x$.  A \emph{clause} is a disjunction of distinct literals. We identify clauses with sets of their literals. A Boolean formula is in the \emph{conjunctive normal form} (CNF) if it is a conjunction of clauses. As an example, all of the following formulas are in CNF:
\begin{equation*}
  x \vee y, \quad x \wedge y, \quad (\overline{x} \vee y) \wedge \overline{z}, \quad (x \vee \overline{y}) \wedge (\overline{x} \vee y \vee z) \wedge (\overline{x} \vee \overline{z}).
\end{equation*}
Given a Boolean formula $F$, an \emph{$F$-assignment} is a set of its literals that does not contain both a literal and its negation. We call an $F$-assignment $\mathbf{T}$ \emph{true} if every clause of $F$ contains a literal from $\mathbf{T}$; $F$ is called \emph{satisfiable} if there is an $F$-assignment. The \emph{CNF-satisfiability problem} is the set $\mathbf{CNF\text{-}SAT}$ of satisfiable Boolean formulas in CNF. Our main result is the following theorem.

\begin{theorem} \label{T:1}
There is an algorithm that, for a given Boolean formula $F$ in CNF, stops in time bounded by $O\left(|F|^4\right)$ and answers yes together with a true $F$-assignment if $F$ is satisfiable and no otherwise.
\end{theorem}

As a corollary, we have that $\mathbf{CNF\text{-}SAT} \in \mathbf{P}$,  where $\mathbf{P}$ is the class of sets of strings recognizable by a deterministic Turing machine in time bounded by a polynomial in the length of the input.


\section{Proofs of statements} \label{S:3}

Given a Boolean formula $F$ in CNF, set of literals $\mathbf{S}$ and literal $z$. Let $[\mathbf{S}]_F$ be the set of clauses of $F$ containing literals from $\mathbf{S}$, $\langle\mathbf{S}\rangle_{F,z}$ the set of literals $u \neq \overline{z}$ such that $[\{u\}]_F \nsubseteq [\mathbf{S}]_F$ and $[\{\overline{u}\}]_F \subseteq [\mathbf{S}]_F$, $\mathbf{S}_F^0(z) = \{z\}$ and $\mathbf{S}_F^{k+1}(z) = \mathbf{S}_F^k(z) \cup \langle\mathbf{S}_F^k(z)\rangle_{F,z}$ for all $k \geq 0$.

Define $\mathbf{S}_F(z) = \mathbf{S}_F^{|F|}(z)$ and $\mathbf{C}_F(z) = [\mathbf{S}_F(z)]_F$, where $|F|$ is the number of symbols contained in $F$. As an example, consider the following Boolean formula in CNF
\begin{equation*}
  (x_1 \vee \overline{x_2} \vee \overline{x_3}) \wedge (x_2 \vee x_3 \vee \overline{x_4}) \wedge (\overline{x_3} \vee x_4) \wedge (\overline{x_1} \vee \overline{x_4}) \wedge x_4,
\end{equation*}
then $\mathbf{S}_F(x_1) = \{x_1, x_2, \overline{x_3}\}$ and $\mathbf{C}_F(x_1) = \{x_1 \vee \overline{x_2} \vee \overline{x_3}, x_2 \vee x_3 \vee \overline{x_4}, \overline{x_3} \vee x_4\}$. It is not hard to check that $\mathbf{S}_F$, $\mathbf{C}_F$ are computable in time bounded by $O\left(|F|^2\right)$ and satisfy the following conditions:
\begin{equation*}
  \text{(1)} \quad u \in \mathbf{S}_F(z) \ \Longrightarrow \ \overline{u} \notin \mathbf{S}_F(z); \qquad
  \text{(2)} \quad u \neq z,\ [\{u\}]_F \subseteq \mathbf{C}_F(z) \ \Longrightarrow \ [\{\overline{u}\}]_F \subseteq \mathbf{C}_F(z);
\end{equation*}
for all literal $u$ of $F$.

A literal $z$ is called \emph{redundant} in $F$ if $[\{\overline{z}\}]_F \subseteq \mathbf{C}_F(z)$. We call $F$ \emph{reduced} if it does not contain redundant literals. Denote by $\rho_z(F)$ the Boolean formula in CNF obtained from $F$ by removing all clauses from $\mathbf{C}_F(z)$. Prove auxiliary lemmas.

\newpage

\begin{lemma} \label{L:1}
If $z$ is redundant, then $\rho_z(F)$ is satisfiable iff $F$ is satisfiable.
\end{lemma}
\begin{proof}
It is sufficient to prove that $\rho_z(F)$ is satisfiable implies $F$ is satisfiable. Let $\rho_z(F)$ is satisfiable and $\mathbf{T}$ is a true $\rho_z(F)$-assignment. If $\mathbf{S}_F(z) = \left\{x_1^{a_1}, \ldots, x_n^{a_n}\right\}$, then $x_1, \ldots, x_n$ are pairwise distinct by~(1) and do not occur in $\rho_z(F)$ by the definition of $\mathbf{C}_F(z)$. Therefore, $\mathbf{T} \cup \mathbf{S}_F(z)$ is a true $F$-assignment and so $F$ is satisfiable.
\end{proof}

\begin{lemma} \label{L:2}
If $F$ is reduced, then $F$ is satisfiable iff $F$ does not contain clauses.
\end{lemma}
\begin{proof}
By induction on the number $N_F$ of variables in $F$. If $N_F = 0$, the statement holds. Let the induction assumption holds for all formula with less than $N_F > 0$ variables, prove it for $F$. Consider a variable $x$ of $F$ and let $C^a_1 \vee x^a, \ldots, C^a_{n_a} \vee x^a$ are the clauses of $[\{x^a\}]_F \setminus [\{\overline{x^a}\}]_F$ for some $n_a \geq 0$. Define a Boolean formula $G$ obtained from $F$ by removing all clauses $C^a_i \vee x^a$  and adding clauses $C^0_i \vee C^1_j$ for all $i$ and $j$. Since $F$ is reduced, $n_0 + n_1 > 0$ and so $G$ contains at least one clause.

First, prove that $[\mathbf{S}_G^k(z)]_F \subseteq \mathbf{C}_F(z)$ for all literal $z$ by induction on $k \geq 0$. If $k = 0$, then $\mathbf{S}_G^k(z) = \{z\}$ and so the statement holds. Let the statement holds for some $k \geq 0$ and $u \in \langle\mathbf{S}_G^k(z)\rangle_{G,z}$. By the definition, $u \neq \overline{z}$, $[\{u\}]_G \nsubseteq [\mathbf{S}_G^k(z)]_G$ and $[\{\overline{u}\}]_G \subseteq [\mathbf{S}_G^k(z)]_G$. Assume that $[\{u\}]_F \nsubseteq \mathbf{C}_F(z)$. Then $[\{\overline{u}\}]_F \nsubseteq \mathbf{C}_F(z)$ by~(2). Consider a clause $C \in [\{\overline{u}\}]_F \setminus \mathbf{C}_F(z)$. If $C \in [\{\overline{u}\}]_G$, then $C \in [\mathbf{S}_G^k(z)]_G$. Therefore, $C \in [\mathbf{S}_G^k(z)]_F \subseteq \mathbf{C}_F(z)$, which is impossible. Otherwise, $C = C^a_{i_a} \vee x^a$ for some $a$, $i_a$ and so $C^0_{i_0} \vee C^1_{i_1} \in [\{\overline{u}\}]_G \subseteq [\mathbf{S}_G^k(z)]_G$ for all $i_{\overline{a}}$. Then either $C \in [\mathbf{S}_G^k(z)]_F \subseteq \mathbf{C}_F(z)$, or $[\{x^{\overline{a}}\}]_F \subseteq [\mathbf{S}_G^k(z)]_F \subseteq \mathbf{C}_F(z)$ and so $C \in [\{x^a\}]_F \subseteq \mathbf{C}_F(z)$ by~(2), which is impossible. Therefore $[\{u\}]_F \subseteq \mathbf{C}_F(z)$ and so $[\mathbf{S}_G^{k+1}(z)]_F \subseteq \mathbf{C}_F(z)$. Particularly, we have $[\mathbf{S}_G(z)]_F \subseteq \mathbf{C}_F(z)$ for all literal $z$.

Next, assume that $G$ contains a redundant literal $z$. Then $[\{\overline{z}\}]_G \subseteq \mathbf{C}_G(z)$ and $[\{\overline{z}\}]_F \nsubseteq \mathbf{C}_F(z)$. Let $C \in [\{\overline{z}\}]_F \setminus \mathbf{C}_F(z)$. If $C \in [\{\overline{z}\}]_G$, then $C \in [\mathbf{S}_G(z)]_G$. Therefore, $C \in [\mathbf{S}_G(z)]_F \subseteq \mathbf{C}_F(z)$, which is impossible. Otherwise, $C = C^a_{i_a} \vee x^a$ for some $a$, $i_a$ and so $C^0_{i_0} \vee C^1_{i_1} \in [\{\overline{z}\}]_G \subseteq [\mathbf{S}_G(z)]_G$ for all $i_{\overline{a}}$. Then either $C \in [\mathbf{S}_G(z)]_F \subseteq \mathbf{C}_F(z)$, or $[\{x^{\overline{a}}\}]_F \subseteq [\mathbf{S}_G(z)]_F \subseteq \mathbf{C}_F(z)$ and so $C \in [\{x^a\}]_F \subseteq \mathbf{C}_F(z)$ by~(2), which is impossible. Therefore, $G$ is reduced and by the induction assumption $G$ is not satisfiable.

Finally, if there is a true $F$-assignment $\mathbf{T}$, then for all $i$, $j$ either $C^0_i$ or $C^1_j$ contains a literal from $\mathbf{T}$. So, $\mathbf{T}$ is a true $G$-assignment, which is impossible. Thus, $F$ is not satisfiable.
\end{proof}

Now describe the following algorithm. Let $\mathbf{T} = \emptyset$. While $F$ contains a redundant literal $z$, reduce $F$ to $\rho_z(F)$ and extend $\mathbf{T}$ by adding literals from $\mathbf{S}_F(z)$. If $F$ is reduced, answers yes together with $\mathbf{T}$ if $F$ does not contain clauses and no otherwise.

Note that $|\rho_z(F)| < |F|$ for all redundant literal $z$, every step can be performed in time bounded by $O\left(|F|^3\right)$, and the obtained set $\mathbf{T}$ is an $F$-assignment such that every removed clause of $F$ contains a literal from $\mathbf{T}$. Due to Lemmas~\ref{L:1} and~\ref{L:2}, the algorithm stops in time bounded by $O\left(|F|^4\right)$ and answers yes together with a true $F$-assignment $\mathbf{T}$ if $F$ is satisfiable and no otherwise. Theorem~\ref{T:1} is proved.



\end{document}